\documentclass[12pt,draftclsnofoot,onecolumn]{IEEEtran}

\usepackage{amsmath,amssymb,amsfonts,amsthm,cite,color}
\usepackage[hidelinks=true]{hyperref}

\usepackage[pdftex]{graphicx}

\usepackage{algorithm}
\usepackage{algorithmic}
\usepackage{bm}            
\usepackage[caption=false,font=footnotesize]{subfig}

\newcommand{\Z}{\mathbb{Z}}
\newcommand{\EX}{\mathbf{E}}
\newcommand{\PR}{\mathbf{P}}

\newcommand{\R}{\mathbb{R}}

\newtheorem{theorem}{Theorem}

\newtheorem{lemma}{Lemma}
\newtheorem{proposition}{Proposition}

\newtheorem{remark}{Remark}

\def\mc{\mathcal}
\def\mbf{\mathbf}
\def\UncertainPi{\textsc{Noback}}
\def\SSum{\textsc{SubsetSum}}
\def\ConcMin{\textsc{ConcMin}}

\def\BufferBanditMathMode{i\textsc{Festival}} 
\def\BufferBandit{{\em i}\textsc{Festival}}
\def\SelectUsers{\textsc{SelectUsers}}
\def\AllocateChannels{\textsc{AllocateChannels}}
\def\BlockVariable{\textsf{b}}

\def\algo#1{\ifnum#1=1
  		{\text{{\sc step\_flex}}}      
               \else
          \ifnum#1=2 
           {\text{{\sc step\_semiflex}}}
          \else
	{\text{{\sc step\_inflex}}}
       \fi\fi}

\IEEEoverridecommandlockouts 

\begin{document}

\title{Multi-channel Resource Allocation for Smooth Streaming: Non-convexity and Bandits \thanks{Akhil Bhimaraju is with the Dep. of Electrical and Computer Engineering, University of Illinois at Urbana-Champaign, Urbana, IL 61801, USA. Email: akhilb3@illinois.edu.} \thanks{Atul A. Zacharias is with the Whiting School of Engineering, Johns Hopkins University, Baltimore, MD 21218, USA. Email: atulantony1998@gmail.com.} \thanks{Avhishek Chatterjee is with the Dept. of Electrical Engineering, Indian Institute of Technology Madras, Chennai, TN 600036, India. Email: avhishek@ee.iitm.ac.in.}}

\author{Akhil Bhimaraju, Atul A. Zacharias and Avhishek Chatterjee}
\maketitle

\begin{abstract}
User dissatisfaction due to buffering pauses during streaming is a significant cost to the system, which we model as a non-decreasing function of the frequency of buffering pause.
Minimization of total user dissatisfaction in a multi-channel cellular network  leads  to a non-convex problem. 
Utilizing a combinatorial structure in this problem, we first propose a polynomial time joint  admission control and channel allocation algorithm which is provably (almost) optimal. This scheme assumes that the base station (BS) knows the multimedia frame statistics of the streams. In a more practical setting, where these statistics are not available a priori at the BS, a learning based scheme with provable guarantees is developed. This learning based scheme has relation to regret minimization in multi-armed bandits with non-i.i.d. and delayed reward (cost). All these algorithms require none to minimal feedback from the user equipment to the base station regarding the states of the media player buffer at the application layer, and hence, are of practical interest. 
\end{abstract}
\section{Introduction}
\label{sec:intro}
Frequent buffering pauses (a.k.a. playout stalls) during multimedia streaming is a source of great dissatisfaction for cellular users. As multimedia is the most significant part of internet traffic today, operators must strive to provide a smooth streaming experience. During video or multimedia streaming, data  transmitted by the base station (BS) are first cached in the media player buffer at the application layer. From this, the media player consumes (plays) one multimedia frame at a time at a rate dictated by the quality,  encoding scheme and dynamics of the content. Whenever the buffer does not have enough data to play the current frame, there is a pause. 

In this work, we address user dissatisfaction due to buffering pause in a multi-channel cellular network. Our formulation captures buffering pause using queuing models for the media player buffers and user dissatisfaction as a function of the frequency of pause. Unlike the traditional stochastic network optimization setting \cite{Neely2010},  this formulation leads to cost-minimization problems with non-convex structures. 
Exploiting combinatorial structure inside the apparently continuous non-convex problem, we develop near optimal resource allocation algorithms. We consider both the scenarios, where the BS knows and where the BS does not know the statistics of the streams a priori. The latter case has connections  to multi-armed bandits with non-i.i.d. and delayed cost. Our proposed algorithms require little to no feedback from the user equipment regarding the buffer states and are compatible with the current cellular implementations.

\subsection{Related literature}
There is a rich body of work on real time scheduling \cite{HouBK2009,HouK2010,JaramilloS2010,LiEL2013,KimLM2014,HouS2016}.
Recently there have been many works on age of information which develop scheduling policies to ensure freshness of the received information in applications like real-time sensing and internet of things \cite{KaulYG2012,CostaCE2014,KamKNE2015,SunUYKS2016,HuangM2015,JiangZNY2019,KadotaSM2018}.   

Dutta et al. \cite{DuttaSACKK2012} and Bhatia et al. \cite{BhatiaLNS2014} studied resource allocation to mitigate pause by utilizing the media player buffers. Dutta et al.  greedily maximized a surrogate, the minimum expected `playout lead' at each scheduling epoch. 
Hou et al. \cite{HouP2015} showed that in a single channel, underloaded network, it is possible to take the frequencies of pause  to zero and also characterized their diffusion limits. 
Xu et al. \cite{XuAltman} analyzed buffer starvation statistics under different service and frame consumption statistics.
Singh et al. \cite{SinghK2015} formulated the problem of minimizing frequency of pause as a Markov decision process and derived a threshold policy. 
This was further extended to obtain a decentralized policy for a distributed network  \cite{SinghK2019}.

In spirit, our work shares most similarity with \cite{HouP2015,SinghK2015,SinghK2019}, which aim to directly address the issue of buffering pause in a single-channel network using a queuing model for the media player buffer. However, there are many differences between those and the current work, some of which are discussed next.
\begin{itemize}
\item The modern cellular networks use  OFDMA and are often overloaded either due to high user density and shadowing in urban areas or low BS density and high pathloss in rural areas. So, in contrast to \cite{HouP2015,SinghK2015,SinghK2019}, our model captures a (possibly) overloaded multi-channel system. 
\item As it is impossible to take the frequency of pause for each user to zero in an overloaded network, we aim to minimize the total user dissatisfaction. Each user's dissatisfaction is modeled as a non-decreasing function of their respective frequency of pause and captures user expectations, which may depend on their data plan, the type of content, and personal factors.
\item The buffers at the application layer can easily store a few minutes  of future content. However,  reporting the buffer states from the application layer of the user to the the MAC layer of the BS at regular intervals is resource consuming, and is not provisioned in the current cellular implementations. So, in contrast to \cite{HouP2015,SinghK2015,SinghK2019}, we assume the media buffer to be sufficiently large and design allocation schemes which are either agnostic of buffer states or  access buffer states  infrequently  (with asymptotically  vanishing rate).
\item {\color{black}From the buffer, the player consumes content as multimedia frames (I, P or B) and the number of frames per second (fps) depends on the content. For current multimedia encoding ($40$--$60$ fps), on average one multimedia frame is consumed per $1.5$--$3$ OFDMA frames, and the multimedia and OFDMA frames are not in alignment. Moreover, the amount of data in a frame, more specifically, in P and B frames, varies with scene dynamics.  Thus, in practice, the amount of data consumed per OFDMA frame by the player from the buffer is stochastic. In \cite{HouP2015,SinghK2015,SinghK2019}, periodic frame consumption by the player was assumed. In this work, we move closer to practice by assuming stationary and ergodic consumption processes. }
\end{itemize}

It is known that servers can adjust (degrade) stream resolutions to suit network conditions (congestion, etc.) \cite{VanderSchaarC2011,RejaieYHE2000,ZhangZZW2006,OliveiraKS1998,BhattacharyyaBR2019,GuttermanFG2020}. 
We first study the scenario where all contents are streamed at their lowest resolutions acceptable to the respective users, which are possibly different for different contents and users. (This captures the case where a user refuses to watch a content below a certain resolution.)  Later we show how our algorithms can be adapted to optimally address users' dissatisfaction due to streaming at degraded resolutions. 
Thus, this work addresses both buffering pause and quality degradation, arguably, the two most pressing issues in streaming.

This paper is organized as follows. The system model and the objective are discussed in Sec.~\ref{sec:model}. Resource allocation schemes, their performance guarantees and proof sketches of the main results are presented in Sec.~\ref{sec:knownpi} and \ref{sec:unknownpi}, when stream parameters (statistics) are known and unknown, respectively.  
Further, in Sec.~\ref{sec:noback}, we also discuss the case where the base station does not have
access to (even infrequent) feedback on the consumption process, but knows a prior
on the parameters of the consumption's distribution statistics.
Simulations strengthening the analytical results are reported in Sec.~\ref{sec:simulation}.
Quality degradation is addressed in Sec.~\ref{sec:quality} followed by conclusion in Sec.~\ref{sec:discussion}.
For detailed proofs, please see the appendices at the end of this manuscript.

\section{System Model and Objective}
\label{sec:model}
We consider the time-slotted OFDMA downlink of a cellular base station (BS) with $m$ channels. The BS is streaming multimedia content to $n$ users over these $m$ wireless fading channels.  In time-slot $s \in \{1, 2, \ldots\}$, user $i \in [n]$ can receive $h_{i,j}(s)$ bits on channel $j \in [m]$. We use $[v]$ to denote the positive integers $\{1, 2, \ldots, v\}$.

The BS decides the allocation of channels and time-slots to users in the beginning of an OFDMA frame, which consists of $\mc{E}$ slots. To avoid confusion with media frames, in the rest of this paper, we refer to OFDMA frames as {\em epochs} and media frames as frames. Epochs are indexed by $t$, i.e., epoch $t$ is composed of time-slots
$(t-1)\mc{E}+1\le s \le t \mc{E}$.

{\color{black}We define $\mbf{H}(t)$ to be an $\R^n \times \R^m \times \R^\mc{E}$-valued process with elements $\{h_{i,j}(s): i \in [n], j\in [m], (t-1)\mc{E}+1 \le s \le t \mc{E}\}$. Here $h_{i,j}(s)$ is the amount of data that the BS sends to user $i$ on channel $j$ in time-slot $s$. This depends on the fading state of the channel and the adaptive modulation and coding  (AMC) techniques employed at the physical layer. 
As there are only finite number of modulation schemes available at the BS, $h_{i,j}(s)$ takes values in a finite set.}

The BS is infinitely backlogged, i.e., all of the content to be served to the users is waiting at the BS. 
Once the content has been served by the BS to a user, it is stored in the user's media player buffer, from which every epoch the media player either reads one {\em frame} or none. For each user $i$, the time of consumption of a frame is denoted by the stochastic process $F_i(t) \in \{0,1\}$. Here $F_i(t)=1$ means that the media player at user $i$ consumes one frame during epoch $t$. This process is stationary and ergodic with $\EX[F_i(t)]=p_i \in [0,1]$. Let $D_i^f$ denote the amount of data (in bits) in frame  $f \in \{1, 2, \ldots\}$ of the content streamed to user $i$. For each $i$, $\{D_i^f: f\ge 1\}$ is a stationary and ergodic process.
So, the amount of data required by the media player of user $i$ at epoch $t$ is 
$F_i(t) D_i^{\sum_{\tau=1}^t F_i(\tau)}$, where $D_i^{\sum_{\tau=1}^t F_i(\tau)}:=D_i^f$ for 
$f=\sum_{\tau=1}^t F_i(\tau)$.

Let $Q_i(t)$ be the occupancy (in bits) of the media player buffer of user $i$ at the end of epoch $t-1$ and the amount of content (in bits) delivered to user $i$ by the BS  in epoch $t$ be  $S_i(t)$. 
As the media player consumes either one frame or none at each epoch, the evolution of
the buffer at user $i$ is given by 
\begin{align}
Q_i&(t+1) = Q_i(t) + S_i(t) - F_i(t) D_i^{\sum_{\tau=1}^t F_i(\tau)} \cdot 
              \mbf{1}(F_i(t) D_i^{\sum_{\tau=1}^t F_i(\tau)}\le Q_i(t) + S_i(t)). \nonumber
\end{align}

We say that the media player at user $i$ has {\em paused} at time $t$ if 
$$\mbf{1}(F_i(t) D_i^{\sum_{\tau=1}^t F_i(\tau)} > Q_i(t) + S_i(t)),$$
i.e., the media player attempted to play the $\sum_{\tau=1}^t F_i(\tau)$th frame, but there was not enough data in the buffer. 

We define a resource allocation policy $a$ to be a sequence of maps $\{a^{(t)}\}$ 
such that at each $t$, 
$\{S_i(t): i \in [n]\}=a^{(t)}\left(\{Q_i(\tau): i \in [n]\},\right.$ $\left.\mbf{H}(\tau):1\le\tau\le t\right)$. 
Let $\mc{A}$ be the class of all ergodic policies under which the time average of the system vector $\{Q_i(t),  S_i(t): i \in [n]\}$ has an almost sure limit in $\R_+\cup\{\infty\}$. For any ${a} \in \mc{A}$ we define the asymptotic {\em frequency of pause} for user $i$ as
$$\kappa^{a}_i = \lim_{T\to\infty} \frac{1}{T} \sum_{t=1}^T \mbf{1}(F_i(t) D_i^{\sum_{\tau=1}^t F_i(\tau)} > Q^{a}_i(t) + S^{a}_i(t)) \mbox{ a.s.},$$
where $S^{a}_i(t)$ and $Q^{a}_i(t)$ are the service and the buffer processes under policy $a \in \mc{A}$.

For each user $i$ there is a cost function $V_i: [0,1] \to \R_+$ which captures the user's dissatisfaction as a function of its frequency of pause. The asymptotic cost for user $i$ under policy $a\in \mc{A}$ is given by $V_i(\kappa^a_i)$. Thus, the total asymptotic cost of the $n$-user and $m$-channel system under policy $a$ is $V^{n,m}(a)=\sum_i V_i(\kappa_i^{a})$, where $\kappa_i^{a}$ may possibly depend on the channel statistics.

As our primary objective is to minimize the total user dissatisfaction due to pause, we find an allocation $a\in\mc{A}$ which minimizes the total asymptotic average cost: 
$$\arg \min_{a \in \mc{A}} V^{n,m}(a).$$

In this paper, we use the notations $O(\cdot)$, $o(\cdot)$ and $\Theta(\cdot)$ with their standard meaning \cite{CormenLRS2009}.

\subsection{Practically relevant cost function}
\label{sec:costfunction}
Standard resource allocation problems in wireless networks involve either a minimization of a convex function or a maximization of a concave function. A traditional choice of  cost function along this line would turn the above problem into a convex problem and thus, would offer more tractability. Unfortunately, in this case,  such a choice would be impractical. For choosing the right  cost functions, let us relate to our own experience during multimedia streaming.

By definition, $0\le \kappa_i \le p_i$, because  frequency of pause cannot be more than the frame rate. To understand the nature of the functions, it is better to first look at the two extremes: $\kappa_i=0$ and $\kappa_i=p_i$. Naturally, we must have $V_i(0)=0$ and  $V_i(p_i)>0$ for all $i$. It is also obvious that the cost functions $\{V_i\}$ must be non-decreasing to capture increased dissatisfaction at an increased frequency of pause.  Near $\kappa_i=p_i$, where almost every frame is paused, a slight decrease in $\kappa_i$ would have almost no impact on user's dissatisfaction, which is at saturation. On the other hand, near $\kappa_i=0$, where the streaming experience is smooth, a slight increase in the frequency of pause would annoy the user significantly. This implies that a natural choice for $\{V_i\}$ are monotone increasing functions whose derivatives are non-increasing. Thus, the class of monotone increasing {\em concave} functions is the right choice for cost.

\subsection{Assumptions}
So far, in describing the system model and the objective, we have made some generic assumptions on the dynamics of the media player buffer and the fading process. For analytical tractability and simplicity of exposition, we introduce some structural assumptions.

The following assumption is motivated by the observations made in Sec.~\ref{sec:costfunction} and by analytical tractability.

{ \bf A1:} For each $i$, $V_i$ is a non-decreasing differentiable concave function with $V_i(0)=0$, the derivative at $0$ bounded by $G$, and $V_i(p_i) = V \cdot p_i$ for some positive constant $V$.

{\color{black}Following the existing literature on resource allocation \cite{ZhangJBC2014, FangZCL2016, HouP2015,bodas14,SinghK2015}, we assume that for any $i\in[n]$, $j \in [m]$ and $t$, $h_{i,j}(s)$ are the same  for all $s \in \{(t-1)\mc{E}+1, \ldots t\mc{E}\}$ and are known to the BS at the beginning of epoch $t$.}
{\color{black}Also, as $h_{i,j}(s)$ take finite values, without loss of generality, we normalize all data quantities, including frame size and $h_{i,j}(s)$, by the maximum possible value that $h_{i,j}(s)$ can take.}

{\bf A2:} For $i \in [n]$ and $j \in [m]$, $h_{i,j}(s)$ are the same for all $s \in [(t-1) \mc{E}+1, t \mc{E}]$ and is denoted by $h_{i,j}(t)$. For each $i$ and $j$, $\{h_{i,j}(t): t \in \Z_+\}$ are i.i.d. and $\bar{h}_{i,j}:=\PR(h_{i,j}(t)=1)\ge\bar{h}$ for some $\bar{h}>0$. Also, for each $t$ and $i$, $\{h_{i,j}(t): 1 \le j \le m\}$ are i.i.d. 

{\color{black}This assumption is well justified for low mobility scenarios where an epoch (i.e., an OFDMA frame) is comparable to the channel coherence time. At higher mobility, the assumption is well justified if scheduling epoch is chosen to be an OFDMA sub-frame or a few OFDMA slots.}

{\color{black} When $\mbf{H}(t)$ is not known at the transmitter, the performance upper bound in Theorem~\ref{thm:benchmark} has a natural extension. It can be shown that \ConcMin~followed by a random scheduler achieves that benchmark if fading statistics are the same across all channels. We omit this result, whose analysis is very similar to that of the results presented here, in the interest of space.}

{\color{black} All other analytical works so far assume that the frames are consumed periodically and are of the same size. As we discuss in Sec.~\ref{sec:intro}, this is not the case in practice. We take a step closer to practice by presenting analytical guarantees for the following more general stochastic multimedia frame dynamics.}

{ \bf A3:} For each $i$, $F_i(t) \in \{0, 1\}$ is stationary and ergodic with $\PR(F_i(t)=1)=p_i$, where $p_i$ is of the form $\frac{z_i}{Z}$ for all $i$. Here $Z$ is an integer independent of the system size and $z_i \in [Z]$ for all $i$. For some $\BlockVariable \in \Z_+$, $D_i^f = \BlockVariable \mc{E}$ for all $i$ and $f$.

The GoP structure and the frame rates are encoded in the header of the stream at the application layer. The MAC scheduler at the BS does not have access to these end-to-end application layer parameters. These parameters are generally used by the media player for decoding and playing the stream. But based on certain metadata shared by the higher network layers or the user equipment, the BS may be able to estimate the frame rate and the GoP structure.  In terms of the mathematical model in Sec.~\ref{sec:model} and the above assumptions, these parameters (statistics) are equivalent to $\{p_i\}$.  We study resource allocation in both scenarios: the BS knows and does not know $\{p_i\}$ a priori.

It is apparent that the cost-minimization problem posed here is quite  different from traditional utility optimization problems in communication networks, 
which are generally solved via novel adaptations of convex algorithms, e.g., dual gradient descent (a.k.a. drift plus penalty method) \cite{Neely2010}, heavy ball method \cite{LiuESB2016}, alternating direction method of multipliers \cite{LiuESB2016}. 
Our cost-minimization problem involves minimization of a differentiable concave cost, and hence is a non-convex problem. 
Moreover, the input variables of the cost functions are not data rates, rather frequencies of pause. 
It is not clear how to write the resource constraints directly in terms of frequencies of pause so that we can obtain a suitable static problem \cite{Neely2010}. Hence, the widely used network optimization techniques  cannot be applied here. 

\subsection{A benchmark}
\label{sec:benchmark}

To analytically compare the performance of our proposed resource allocation policies, a benchmark is needed. The following theorem provides a universal benchmark for all ergodic allocation schemes.
\begin{theorem}
\label{thm:benchmark}
Under assumptions {\bf A1-A3}, the cost of any ergodic policy is lower bounded by
\begin{align}
\bar{V}^{(n,m)}=\min_{\{0 \le \alpha_i \le 1\}} \sum_{i=1}^n V_i(\max(p_i-\alpha_i,0)) \ \mbox{s.t.} \sum_i \alpha_i \le \frac{m}{\BlockVariable}.
\label{eq:benchmark}
\end{align} 
\end{theorem}
This bound is applicable for any $\bar{h}>0$ in assumption {\bf A2}, and thus is independent of the fading statistics.
Later, we show comparison of the cost under our proposed policy with this lower bound.  The above theorem follows from the following lemma.
\begin{lemma}
\label{lem:benchmark}
Under assumptions {\bf A1-A3}, for any ergodic policy $a \in \mc{A}$, if the ergodic service rate to user $i$ is $\bar{s}_i^{a}:=\lim_{\tau\to \infty} \frac{1}{\tau}\sum_{t=1}^{\tau} S_i^{a}(t)$, then $\kappa_i^a = \max(p_i - \frac{\bar{s}^a_i}{\BlockVariable \mc{E}}, 0)$.
\end{lemma}

This expression for \(\kappa_i^a\) is obtained by establishing a simple relation between the probability of buffering pause and the expected change in the buffer state at epoch \(t\).
We can see that setting \(\frac{\bar{s}_i^a}{\BlockVariable\mc{E}} = \alpha_i^*\) achieves the lower bound in Thm. \ref{thm:benchmark}, where \(\{\alpha_i^*\}\) are the optimal solutions of \eqref{eq:benchmark}. 
This bound might be achievable in the absence of fading or when the system is underloaded.
However, for an overloaded system, i.e., when \(\sum_i p_i > \frac{m}{\BlockVariable}\), especially in the presence of fading, it is not possible to achieve $\frac{\bar{s}_i^a}{\BlockVariable\mc{E}} = \alpha_i^*$ for all \(i\) simultaneously, since this would otherwise require that \(\sum_i\frac{\bar{s}_i^a}{\BlockVariable\mc{E}} = \frac{m}{\BlockVariable}\), i.e., the total ergodic service rate should not be impacted by fading at all.
Hence, for fading channels, a gap with the benchmark is expected.

\section{Known $\{p_i\}$: non-convexity and joint admission-allocation}
\label{sec:knownpi}

We start with the  case when $\{p_i\}$ are known at the BS a priori, since it is the simpler case which helps  to separate the complexity in cost minimization from the additional challenges due to the lack of knowledge of $\{p_i\}$. 

As discussed in Sec.~\ref{sec:model}, the lack of a convex structure does not allow us to use the traditional network optimization techniques \cite{Neely2010}.
We take an indirect approach which harnesses a combinatorial structure inside the continuous non-convex problem and gives an optimal joint admission control and channel allocation scheme. 

Our approach is motivated by the following simple observation based on Thm.~\ref{thm:benchmark} and Lem.~\ref{lem:benchmark}. If we can find  $\{\alpha^*_i\}$ that solves the optimization problem \eqref{eq:benchmark} and can obtain an allocation scheme $\bar{a}$ such that
$\frac{s^{\bar{a}}_i}{\BlockVariable \mc{E}}=\alpha^*_i$, then $\bar{a}$ is an optimum resource allocation scheme.
Towards this, we develop a polynomial time algorithm \ConcMin~which solves \eqref{eq:benchmark} (Sec.~\ref{sec:ConcMin}) and design a polynomial time channel allocation algorithm \AllocateChannels~under which $\alpha_i^*-\frac{s^{\bar{a}}_i}{\BlockVariable \mc{E}} \le \theta^{-m}$ for $\theta>1$ (Sec.~\ref{sec:AllocateChannels}). 

\begin{algorithm}
\caption{\ConcMin} \label{alg:concMin}
{\bf Input:} \(\{V_i\}, \{p_i\}, c=\frac{m}{\BlockVariable}\)

{\bf Output:} \(\{\bar{\alpha}_i\}\)

\begin{algorithmic}[1]
\IF{\(\sum_{i \in [n]} p_i \le c\)} \label{alg:lin:plec}
    \STATE \(\bar{\alpha}_i \leftarrow p_i\) for all \(i \in [n]\)
\ELSE
    \FORALL{\(k \in [n]\)} \label{alg:lin:findfracuser}
        \STATE \(L_k \leftarrow 
                    \SSum([n] \setminus k, c)\)
                    \label{alg:lin:findlk}
        \STATE \(L \leftarrow V \cdot \left(\sum\limits_{i \in [n]\setminus k} p_i -\sum\limits_{i \in L_k} p_i\right) +
                            V_k(p_k + \sum\limits_{i \in L_k} p_i - c)\)
        \STATE \COMMENT{\(L\) is cost if 
                            \(\alpha_i=p_i\) for \(i \in L_k\)}
        \STATE \(R_k \leftarrow 
                    \SSum([n] \setminus k, \sum\limits_{i \in [n]} p_i -c)\)
                    \label{alg:lin:findrk}
        \STATE \(R \leftarrow V \cdot \left(\sum\limits_{i \in R_k} p_i\right) + 
                            V_k(\sum\limits_{i \in [n]} p_i -\sum\limits_{i \in R_k} p_i-c)\)
        \STATE \COMMENT{\(R\) is cost if 
                            \(\alpha_i=0\) for \(i \in R_k\)}
        \IF{\(L \le R\)}
            \STATE \(\alpha_i^k \leftarrow p_i\) for all \(i \in L_k\)
            \STATE \(\alpha_k^k \leftarrow c - \sum_{i \in L_k} p_i\)
            \STATE \(\alpha_i^k \leftarrow 0\) for all \(i \notin L_k \cup \{k\}\)
            \STATE \(J_k \leftarrow L\)
        \ELSE
            \STATE \(\alpha_i^k \leftarrow 0\) for all \(i \in R_k\)
            \STATE \(\alpha_k^k \leftarrow c - \sum_{i \in [n]\setminus k} p_i + 
                                        \sum_{i \in R_k} p_i\)
            \STATE \(\alpha_i^k \leftarrow p_i\) for all \(i \notin R_k \cup \{k\}\)
            \STATE \(J_k \leftarrow R\)
        \ENDIF \label{alg:lin:endofBestSkstar}
    \ENDFOR
    \STATE \(k^* \leftarrow \arg \min_k J_k\) \label{alg:lin:findkstar}
    \STATE \(\bar{\alpha}_i \leftarrow \alpha_i^{k^*}\) for all \(i \in [n]\) \label{alg:lin:alphabar}
\ENDIF
\end{algorithmic}

\end{algorithm}

\subsection{\ConcMin~for solving \eqref{eq:benchmark}}\label{sec:ConcMin}
\ConcMin~(Alg.~\ref{alg:concMin}) is proposed to solve the optimization problem in Thm.~\ref{thm:benchmark}. In the case of an under-loaded (resource rich) network, i.e., $\sum_{i\in[n]} p_i \le \frac{m}{\BlockVariable}$, $\alpha_i=p_i$ for all $i$ is the obvious optimal solution (Step \ref{alg:lin:plec}).
The main challenge lies in the overloaded or resource constrained network, i.e., $\sum_{i\in[n]} p_i > \frac{m}{\BlockVariable}$. In this case, \ConcMin\ searches over a collection of extreme points of the constraint set and picks one with the minimum cost. Here the extreme points are the set of tuples $\{\alpha_i: i \in [n]\}$ such that for some $\mc{S}\subset [n]$ and $|\mc{S}|=n-1$, $\alpha_i \in \{0, p_i\}$ for all $i \in \mc{S}$. This search is carried out in Steps \ref{alg:lin:findfracuser}-\ref{alg:lin:alphabar}.

To find the best extreme point, for each $k \in [n]$, \ConcMin~ searches for the subset $\mc{S}_k^* \subset [n]\setminus k$ and the best $\alpha_k \in (0, p_k)$  so that if $\alpha_i = p_i$ for $i \in \mc{S}_k^*$ and $\alpha_i=0$  for  $i \not\in \mc{S}_k^*\cup\{k\}$, then the cost is minimized ({\bf for} loop in Step \ref{alg:lin:findfracuser}). Finally, it picks the best $k$ and the corresponding $\mc{S}_k^*$ by comparing cost of $\{\mc{S}_k^*: k\in [n]\}$ (Steps \ref{alg:lin:findkstar}-\ref{alg:lin:alphabar}). 

The search for $\mc{S}_k^*$ is a combinatorial subset selection problem. \ConcMin~finds  $L_k$ which maximizes $\sum_{i \in S\setminus k} p_i$ and $R_k$ which minimizes $\sum_{i \in S\setminus k} p_i$ subject to $\sum_{i \in S\setminus k} p_i > 1 - p_k$. The one with lower cost among them is picked as $S^*_k$.
Finding $L_k$ is related to the well known subset sum problem (Step \ref{alg:lin:findlk}) \cite{KellererPP2004}. It turns out that the problem of finding $R_k$ can be written in an alternate form, which is also a subset sum problem with different parameters (Step \ref{alg:lin:findrk}).

We use the \SSum~routine to solve the subset sum problem. \(\SSum(W,c)\), for some \(W \subseteq [n]\), returns the set \(S \subseteq W\) 
so that \(\sum_{i \in S}p_i\) is maximized subject to \(\sum_{i \in S}p_i \le c\).
For \SSum~the standard dynamic programming based algorithm \cite{KellererPP2004} can be used. Though  that algorithm does not solve any general subset sum problem in polynomial time, in our case it does. This is because, for our problem, across all instances the sack sizes are at most $Z\cdot\max(m,n)$.  Further, as subset sum  is a special case of the knapsack problem and the weights $\{p_i\} \subset \{\frac{z}{Z}: z \in [Z]\}$ for $Z=O(1)$, there exists an accurate algorithm with $O(n)$ complexity \cite{KellererPP2004}.

We have the following guarantee on the computational complexity and the correctness of \ConcMin.

\begin{theorem}
\label{thm:concmin}
In $O(n^2 )$ steps \ConcMin~obtains an optimal solution for the optimization problem in Thm.~\ref{thm:benchmark}, i.e., $\bar{\alpha}_i=\alpha_i^*$ for all $i \in [n]$.
\end{theorem}

Interestingly,  the optimization problem in Thm.~\ref{thm:benchmark} involves continuous variables with no apparent integer or combinatorial constraints. However, the particular non-convex structure of the problem leads to an optimal combinatorial algorithm. 

For simplicity of presentation we restrict to \(\BlockVariable=1\) in the following, though all the results directly extend to \(\BlockVariable \ge 1\).

\subsection{Near-optimal channel allocation}\label{sec:AllocateChannels}
Our near-optimal channel allocation scheme \AllocateChannels~is preceded by a sub-routine \SelectUsers~which generates a  list of \(m\) random users. It is ensured that the list does not have a user repeated more than twice. Otherwise, the system would fail to harness the diversity in the fading processes experienced by different users. 
It is also ensured that each user $i$ is picked with probability $\alpha_i^*.\ 
$\SelectUsers\ proceeds as follows.

We maintain a unit length interval for each of the \(m\) ``slots'' we are going to fill with users.
Fill in all the unit length intervals in sequence, starting with \(\alpha_1\) for user \(1\), all the way up to \(\alpha_n\) for user \(n\).
If any \(\alpha_i\) overflows the interval of a particular slot, fill in the remainder of that \(\alpha_i\) in the next slot.
This procedure is pictorially represented in Fig.~\ref{fig:SelectUsers}.
In an overloaded network, we have \(\sum_i \alpha_i^* = m\), which allows us to fill in all the slots perfectly.
On the other hand, in an underloaded network, we have \(\alpha_i^* = p_i\) for all \(i\) with \(\sum_i\alpha_i^* \le m\), and we can leave the last few slots of \SelectUsers\ vacant.
For each slot, we pick user \(i\) with probability equal to the amount filled in by \(\alpha_i^*\) for that slot.
It is easy to see that the expected number of slots allotted 
by \SelectUsers\ to user \(i\) is \(\alpha_i^*\), 
and the maximum number of slots allotted to any user is $2$.

\begin{figure}[h]
    \centering
    \includegraphics[width=0.9\linewidth]{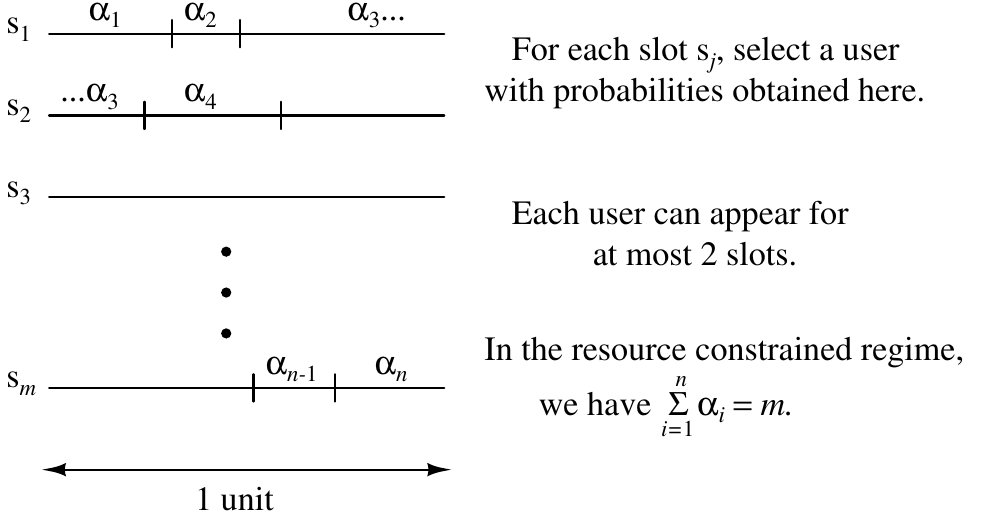}
    \caption{Selecting a set of \(m\) users.}
    \label{fig:SelectUsers}
\end{figure}

\begin{algorithm}[h]
\caption{\SelectUsers} \label{alg:SelectUsers}
{\bf Input:} \(\{\alpha_i\}_{i=1}^n\)

{\bf Output:} List of \(m\) (possibly repeated) users \((s_1, \ldots, s_m)\)

\begin{algorithmic}[1]
\STATE \(u \leftarrow 1\)
\FOR {\(j := 1\ \TO\ m\)}
    \STATE \(Q \leftarrow \emptyset\)
    \WHILE {\(\left(\sum\limits_{\beta:(\nu,\beta)\in Q}\beta\right) < 1\)}
        \IF {\(\alpha_u \le 1 - \left(\sum\limits_{\beta:(\nu,\beta)\in Q}\beta\right)\)}
            \STATE \(Q \leftarrow Q \cup \{(u, \alpha_u)\}\)
            \STATE \(u \leftarrow u + 1\)
        \ELSE
            \STATE \(Q \leftarrow Q \cup \{(u, 1 - \sum_{\beta:(\nu,\beta)\in Q}\beta)\}\)
            \STATE \(\alpha_u \leftarrow \alpha_u - (1 - \sum_{\beta:(\nu,\beta)\in Q}\beta)\)
        \ENDIF
    \ENDWHILE
    \STATE {// Now select \(\nu\) with probability \(\beta\)}
    \STATE {Pick \(Y_\nu \sim \exp(\beta)\) for all \((\nu, \beta) \in Q\) independently}
    \STATE \(s_j \leftarrow \arg\max\limits_{\nu:(\nu,\beta)\in Q} Y_\nu\)
\ENDFOR
\end{algorithmic}
\end{algorithm}

Once we have a list of \(m\) users, we create a bipartite graph between these users and the \(m\) channels, with an edge if and only if the {\color{black}channel rate $h_{i,j}(t)=1$ for user $i$ and channel $j$}. 
\AllocateChannels\ constructs this bipartite graph, finds a maximum matching, and then allocates the matched channels to the users. A maximum bipartite matching can be found using existing algorithms
 \cite{hopcroftkarp,CormenLRS2009}. 
 As no user is repeated more than twice in the list \((s_1, \ldots, s_m)\), a perfect matching is found with very high probability due to the diversity in the fading processes across different users.

\begin{algorithm}[h]
\caption{\AllocateChannels} \label{alg:AllocateChannels}
{\bf Pre-computation at time \(0\):} obtain \(\{\alpha_i^*\}\) from \ConcMin

\begin{algorithmic}[1]
\STATE {For each epoch \(t\):}
\STATE {Get channel state information \(\{h_{i,j}\}\) for all \(i \in [n]\) and \(j \in [m]\)}
\STATE \((s_1, \ldots, s_m) \leftarrow \SelectUsers(\{\alpha_i^*\})\)
\STATE {Construct a bipartite graph between the selected users \(\{s_u\}_{u=1}^m\) and channels \([m]\) as follows: there is an edge between \(s_u\) and \(j\) iff \(h_{s_u,j}=1\)}
\STATE {Find the subset of edges \(M\) which forms a maximum matching of the bipartite graph}
\STATE {For each \((s,j)\in M\), allocate channel \(j\) to user \(s\)}
\end{algorithmic}
\end{algorithm}

{\color{black}\AllocateChannels~allocates a channel to a user only if it is the best, i.e., it allocates channels to users only if  $h_{i,j}(t)=1$. Despite the conservative allocation, it is able to exploit the diversity across  channel-user pairs to allocate resources in an almost optimal fashion, as}
stated formally in the following theorem.
\begin{theorem}
\label{thm:fading}
Under assumptions {\bf A1-A3}, for sufficiently large \(m\), \AllocateChannels\ has a cost \(V^{n,m}(\AllocateChannels)\) that satisfies
\begin{align*}
V^{n,m}(\AllocateChannels) - \bar{V}^{n,m} \le \gamma^{-m},
\end{align*}
for some constant \(\gamma > 1\), for any \(n\).
The per epoch computational complexity of running \AllocateChannels\ is \(O(m n + m^{2.5})\). 
\end{theorem}

The probability that a user does not receive the highest AMC rate on any of the $m$ channels is lower bounded by $(1-\bar{h})^{m}$, for $0<\bar{h}<1$.  Hence,  for $0<\bar{h}<1$, an $\Omega\left((1-\bar{h})^{m}\right)$ loss in per-user throughput compared to the no fading case is unavoidable. 
The benchmark in Theorem~\ref{thm:benchmark} is a universal lower bound, applicable even to the case $\bar{h}=1$. Thus, for $\bar{h}<1$, there will be a gap between the benchmark and the cost incurred by any policy, including the optimal policy.
\AllocateChannels~guarantees that this gap decays exponentially with the number of channels, which is no slower than the decay of per-user throughput loss, and hence, seems to be order optimal.

The main part of the proof of this theorem requires us to show that a user for which \ConcMin\ allocates $\alpha^*_i>0$ is served at any time $t$ with probability at least $\alpha_i^*-\theta^{-m}$ for some $\theta>1$.
For this, we extend  \cite[Lem.~1]{bodas14} to the case where a user's channel states are correlated with  other $O(1)$ users.
Rest follows using Lem.~\ref{lem:benchmark} in Sec.~\ref{sec:benchmark}.

\begin{remark}
\label{rem:AMC}
The proposed algorithm and its performance guarantees are agnostic of the particular AMC technique. For any given AMC technique, the proposed (almost) minimizes the total cost due to buffering pauses for that AMC technique.
\end{remark}

\section{Unknown $\{p_i\}$}
\label{sec:unknownpi}
As discussed in Sec.~\ref{sec:model}, $\{p_i\}$ are application layer parameters and hence, not always known to the MAC scheduler of the BS a priori. Moreover, two videos with the same quality (i.e., HD, 4k) can have different $\{p_i\}$ depending on their dynamism, e.g., sports versus news. Hence, even the application layer may not know accurate values of $\{p_i\}$ a priori. {\color{black} To the best of our knowledge, none of the prior analytical works on streaming has addressed this issue.}

In practice, all multimedia sessions are of finite duration. Hence, it is also important that the allocation scheme performs well not only in terms of the asymptotic average cost, but also in terms of average cost over any reasonable  time window. Further, as discussed before, any implementable allocation scheme at the BS can at best have an infrequent feedback regarding users' media player buffers.

Under a policy $a\in \mc{A}$, let $\kappa^a_i(T)$ be the empirical frequency of pauses over $T$ epochs. Ideally, we should have a policy $a$ with low $\sum_{i \in [n]} V_i(\kappa^a_i)$ and low $\sum_{i \in [n]} V_i(\kappa^a_i(T))$ for all $T$. More precisely, if $\bar{\mc{A}}$ is the class of ergodic policies which minimize asymptotic average cost, ideally, we would like to have the policy $a^* \in \bar{\mc{A}}$, if it exists, such that for all sufficiently large $T$ and any $\bar{a} \in \bar{\mc{A}}$,
\begin{align}\sum_{i \in [n]} \EX\left[V_i(\kappa^{a^*}_i(T))\right] \le \sum_{i \in [n]} \EX\left[V_i(\kappa^{\bar{a}}_i(T))\right]. \label{eq:bandit1}
\end{align}
Note that for all sufficiently large $T$, $\bar{V}^{n,m}$ is still  a benchmark for 
$\sum_{i \in [n]} \EX\left[V_i(\kappa^{a}_i(T))\right]$. Hence,  
\eqref{eq:bandit1} is equivalent to finding $\bar{a} \in \bar{\mc{A}}$ for which  $$v(\bar{a},T) := \sum_{i \in [n]} \EX\left[V_i(\kappa^{\bar{a}}_i(T))\right] - \bar{V}^{n,m}$$ is minimum for all sufficiently large $T$. 
Clearly, for any $\bar{a} \in \bar{\mc{A}}$ as $T\to\infty$, $v(\bar{a},T)\to 0$. As the above multi-objective problem is intractable, we look for a policy under which $v(\bar{a},T)$ rapidly goes to $0$ as $T\to \infty$.

\subsection{Infrequent buffer feedback and bandits}
\label{sec:bandit}
Using Jensen's inequality to move the expectation inside $V_i$ and then using concavity of $V_i$ and assumption {\bf A1}, it follows that
$$v(\bar{a},T) \le G \sum_{i \in [n]} \max(\EX[\kappa^{\bar{a}}_i(T)] - \kappa^{\bar{a}}_i, 0).$$
Thus, for upper bounding the rate of decay of $v(\bar{a},T)$ it is sufficient to upper-bound the rate of decay of $\EX[\kappa^{\bar{a}}_i(T)] - \kappa^{\bar{a}}_i$ for each $i$. Let  $\psi^{\bar{a}}_i(T)$ denote the number of pauses for user $i$ over $T$ epochs under policy $\bar{a}$. Then, upper-bounding the rate of decay of $\EX[\kappa^{\bar{a}}_i(T)] - \kappa^{\bar{a}}_i$ is equivalent to upper-bounding the rate of growth of
$\EX[\psi^{\bar{a}}_i(T)]$.

{ It may be tempting to use the following simple approach. At the beginning, the  user estimates $p_i$ by observing the evolution of the media player buffer for some time and reports it to the BS, then the BS uses \AllocateChannels. Though this is a possible approach, it is sub-optimal. This is because for the above estimation steps, the buffers of the users need to have enough frames, for which the BS needs to transmit sufficient contents to all the users. However, during this estimation period, transmissions to the users who are not part of the optimal schedule in \ConcMin\ are in a sense wasted, which could have been used to improve experience of the other users. This implies that we need to strike a balance between exploration and exploitation.}

This naturally brings us to the setting of multi-armed bandits \cite{BubeckC2012} with non-i.i.d. cost (instead of reward), where a cost of $1$ is incurred for a user every time its stream is paused. The cost is non-i.i.d. because the cost depends on the past states of the buffer, even when $\{F_i(t)\}$ are i.i.d. Moreover, for an action taken at time $t$, the cost may be incurred at a later time. Though there is a similarity in terms of the non-i.i.d. nature of the system, the dynamics and the costs in this problem are different from the queuing bandits studied in \cite{KrishnasamySJS2016,CayciE2017,KrishnasamyAJS2018}. 

Drawing intuition from the bandit literature \cite{BubeckC2012,KrishnasamySJS2016,CayciE2017,KrishnasamyAJS2018} and the analysis of \ConcMin~and \AllocateChannels, we develop an algorithm called {\em i}nfrequent \textsc{F}eedback, \textsc{est}imate, sol\textsc{v}e, and \textsc{al}locate (\BufferBandit), which takes infrequent one bit feedback about the buffer states, estimates $\{p_i\}$ based on that, and allocates using \ConcMin~and \AllocateChannels. 

\begin{algorithm}
\caption{\BufferBandit}
\label{alg:bufferbandit}
{\bf Input:} $\{V_i\}$ and $r, w \in \{2, 3, \ldots\}$\\
{\bf Output:} Allocation at each epoch \\
{\bf Initial computation:} Define phases $\tau=1, 2, \ldots$ where $\tau$th phase consists of epochs
$(\tau-1) (w+1)\lceil \frac{n \BlockVariable}{m}\rceil+1$ to $\tau(w+1)\lceil \frac{n \BlockVariable}{m}\rceil$

\begin{algorithmic}[1]
\WHILE{System is ON}
\IF{for some $q\in\Z_+\cup\{0\}$, current phase $\tau=r^q$}
\STATE Between epochs $(\tau-1) (w+1)\lceil \frac{n \BlockVariable}{m}\rceil+1$ to $(\tau-1) (w+1)\lceil \frac{n \BlockVariable}{m}\rceil+ w \lceil \frac{n \BlockVariable}{m}\rceil$: allocate users \BlockVariable\ ON (for that user) channels each in a work conserving round-robin manner (each user is chosen for $w$ epochs and allocated \BlockVariable-channels in each one of them) \label{alg:bufferbandit:1}
\STATE Between epochs $(\tau-1) (w+1)\lceil \frac{n \BlockVariable}{m}\rceil+ w \lceil \frac{n \BlockVariable}{m}\rceil + 1$ to  $\tau (w+1)\lceil \frac{n \BlockVariable}{m}\rceil$: each user sends $\{1,0\}^w$ feedback about increment of $\{Q_i(t)\}$ or not, respectively, in the $w$ epochs they are allocated in Step \ref{alg:bufferbandit:1}
\label{alg:bufferbandit:2}
\STATE For each $i\in [n]$, based on feedback in Step \ref{alg:bufferbandit:2} update  $\hat{p}_i$ by the total number of $0$s received from user $i$ (since $t=1$) divided by $w \cdot q$
\STATE Run \ConcMin~with $\{\hat{p}_i\}$ to obtain $\{\hat{\alpha}_i\}$
\ELSE
\STATE Run \AllocateChannels~with the latest $\{\hat{\alpha}_i\}$
\ENDIF
\ENDWHILE
\end{algorithmic}
\end{algorithm}

\BufferBandit, described in Alg.~\ref{alg:bufferbandit}, divides time into phases of length $(w+1)\lceil \frac{n \BlockVariable}{m}\rceil$ epochs, where $w \in \Z_+$. For $r \in \Z_+$ and $r\ge 2$, at phases $r, r^2, r^3, \ldots$, \BufferBandit~serves each user in turn over $\BlockVariable$ channels of an entire epoch and the users record the change (increase or same) of their buffer states at the end of that epoch. In each phase this is done $w$ times in  a round-robin fashion over the first $w\lceil \frac{n \BlockVariable}{m}\rceil$ epochs of this phase. From the $(w\lceil \frac{n \BlockVariable}{m}\rceil+1)$th epoch to $(w+1)\lceil \frac{n \BlockVariable}{m}\rceil$th epoch of this phase, the BS collects all the $w$ one bit feedback regarding change of buffer states. Based on this feedback, it estimates $\{p_i\}$ and runs \ConcMin~with these estimates. For any $q \in \Z_+$ between phases $r^q$ and $r^{q+1}$, \BufferBandit~runs \AllocateChannels~with $\{\bar{\alpha}_i\}$ returned by \ConcMin~run during phase $r^q$.

As \BufferBandit~collects only infrequent  feedback ($\frac{w \log T}{T \log r}$ bits per epoch) from the user equipment, it can be implemented in practice for multimedia streaming in cellular networks. Also, feedback from each user is scheduled a priori (at particular epochs in phases $1, r, r^2, \ldots$) and hence, the uplink traffic due to the feedback is well regulated.

For \BufferBandit, we have the following guarantee on the growth of the expected number of pauses with the horizon $T$.

\begin{theorem}
\label{thm:bandit}
Under assumptions {\bf A1-A3} and i.i.d. $\{F_i(t)\}$, if $T\ge r^2 (w+1)\lceil \frac{n \BlockVariable}{m}\rceil$, and $w>\frac{2 \ln r}{\min_{i,j} |p_i - p_j|}$,  then in the absence of fading, i.e., $\mbf{H}(t)=\mbf{1}$, 
$$\EX[\psi^{\text{\BufferBanditMathMode}}_i(T)] \le \max(p_i-\alpha^*_i,0) T + \bar{C} T^{2/3}\log T,$$
for all $i \in [n]$, where $\bar{C}$ is independent of $T$. For any $\mbf{H}(t)$ satisfying {\bf A2}, if $n=\Theta(m)$ then for $\theta=\frac{2}{1+\sqrt{1-\bar{h}}}$  and $C'>0$, a constant independent of $T$,
$$\EX[\psi^{\text{\BufferBanditMathMode}}_i(T)] \le \max(p_i-\alpha^*_i+\theta^{-m},0) T + C'~\theta^{-2m}~\log T.$$
\end{theorem}

Proof of this theorem has two main steps. First, we show that after $t$ epochs, the estimations of $\{p_i\}$, which take values in $\{\frac{z}{Z}: z \in [Z]\}$, are exact with probability at least $1 - \frac{1}{t^{1+\beta}}$ for some $\beta>0$. Second, we bound the expected number of pauses between $r^q$th and $r^{q+1}$th phases assuming the estimate at the end of the $r^q$th phase is accurate. Combining these two along with some standard probability computations gives the result. Proving the first part is a standard application of Azuma-Hoeffding inequality. The second part requires bounding  the expected number of returns to state $0$ by the Markov chain $Q_i(t)$ over a finite time window. Towards that we study a stochastically dominating Markov chain using techniques from \cite{HerveL2014}.

The following result is a consequence of Thm.~\ref{thm:bandit} and the discussions on $v(\bar{a},T)$ in the beginning of Sec.~\ref{sec:bandit}.
\begin{proposition}
\label{prop:bandit}
Under assumptions {\bf A1-A3}, i.i.d. $\{F_i(t)\}$ and $\mbf{H}(t)=\mbf{1}$, if $T\ge r^2 (w+1)\lceil \frac{n \BlockVariable}{m}\rceil$, and $w>\frac{2 \ln r}{\min_{i,j} |p_i - p_j|}$
$$v(\text{\BufferBanditMathMode}, T) = O\left(\frac{\log T}{T^{1/3}}\right).$$
For any $\mbf{H}(t)$ satisfying {\bf A2}, if $n=\Theta(m)$ then for $\theta=\frac{2}{1+\sqrt{1-\bar{h}}}$,
$$v(\text{\BufferBanditMathMode}, T) \le O\left(\theta^{-m}\right) + O\left(\frac{\log T}{T}\right).$$
\end{proposition}

Arguably, the $O\left(\theta^{-m}\right)$  bound on $v(\text{\BufferBanditMathMode}, T)$ is the best that can be achieved in the presence of fading. This is because, as discussed after Theorem~\ref{thm:fading}, even when $\{p_i\}$ are known there exists an exponentially decaying (with $m$) gap between the benchmark and the optimal policy. On the other hand, in the no fading case, $v(\text{\BufferBanditMathMode}, T)$ tends to $0$ since the lower-bound in Theorem~\ref{thm:benchmark} is achievable in this case.

Note that unlike Theorem~\ref{thm:fading}, in Theorem~\ref{thm:bandit} and Proposition~\ref{prop:bandit}, we assume $\{F_i(t)\}$ to be i.i.d., which is required for the analysis. The results can be extended to Markovian $\{F_i(t)\}$ under which $\{Q_i(t)\}$ are geometrically ergodic. However, our present analytical techniques for proving Theorem~\ref{thm:bandit} does not seem to extend to general $\{F_i(t)\}$ processes.

\subsection{Without buffer feedback}
\label{sec:noback}
As mentioned before, current protocols do not generally implement a procedure to feedback the buffer states of the media player at the application layer to the BS. Though we show that \BufferBandit~requires infrequent and simple feedback, one may still ask: what is the extra cost (user dissatisfaction) if we do not use any feedback? Clearly, as there is no feedback, there is no scope to employ adaptive schemes which learn and improve. It also turns out that in the absence of any buffer feedback, the asymptotic cost is bounded away from $\bar{V}^{n,m}$, i.e., $\lim_{T\to\infty} v(a,T) \ge \delta>0$. So, in this setting, our goal is to find a stationary policy with the minimum asymptotic average cost. 

In Sec.~\ref{sec:AllocateChannels}, we observed that for any feasible ergodic rate, \AllocateChannels~(almost) achieves the lower bound on the frequency of pause. Hence, in the absence of feedback, the optimal approach would be to find the best ergodic rates and then employ \AllocateChannels.

For each user $i$, let the available information regarding $p_i$ be its cumulative distribution $G_i(\cdot)$. Then the best service rate allocation can be obtained by solving
\begin{align}
\min_{\{0\le\alpha_i\le 1\}} \sum_{i=1}^n \EX_{p_i} V_i(\max(p_i-\alpha_i,0)) \ \mbox{s.t.} \sum_i \alpha_i \le \frac{m}{\BlockVariable}. \label{eq:uncertain1}
\end{align}
In its full generality, the optimization problem in \eqref{eq:uncertain1} is quite challenging, whose numerical solution is sometimes unobtainable. Under some mild conditions on $\{G_i\}$, \eqref{eq:uncertain1}  turns out to be a minimization of sum of $n$ functions, where each function is concave on a part of its domain and convex on the rest. We leave this unique non-convex problem of independent mathematical interest as future work. Here, we restrict ourselves to the special case of linear cost: for each $i$, $V_i(x)= w_i x$ and $p_i \in [a_i,b_i] \subsetneq [0,1]$ for distinct $\{w_i\}$. 

Without loss of generality, let us assume \(w_1 < w_2 <  \ldots < w_n\) and
define \(w_0 = 0\).
\UncertainPi,  which stands for u\textsc{N}known c\textsc{o}nsumption from the \textsc{b}uffer without feedb\textsc{ack} (Alg.~ 
\ref{alg:uncertainpi}), solves \eqref{eq:uncertain1} for linear $\{V_i\}$.

\begin{algorithm}
\caption{\UncertainPi}
\label{alg:uncertainpi}
{\bf Input:} \(\{w_i\}\), \(\{(a_i, b_i)\}\), \(\{G_i\}\), $c=\frac{m}{\BlockVariable}$ \\
{\bf Output:} \(\{\alpha_i^*\}\)

\begin{algorithmic}[1]
\IF {\(\sum_{i=1}^n b_i \le c\)}
    \STATE \(\alpha_i^* \leftarrow b_i\) for all \(i \in [n]\)
\ELSE
    \STATE \(l \leftarrow 0\) \label{alg:uncertainpi:startlambdasearch}
    \REPEAT 
        \STATE \(l \leftarrow l+1\)
    \UNTIL {\(\sum_{i=l}^n G_i^{-1}(1 - w_l/w_i) \le c\)} \label{alg:uncertainpi:untilconditionlambdasearch}
    \STATE \(\lambda^*\! \leftarrow\! \text{solve} \left\{
                        \sum_{i=l}^n G_i^{-1}(1-\lambda/w_i) = c \right\} \mbox{ for } \lambda\!\! \in\!\! (0,w_l]\)
                        \label{alg:lin:solvelambda}
    \IF {\(\lambda^* > w_{l-1}\)} \label{alg:uncertainpi:startsettingalpha}
        \STATE \(\alpha_i^* \leftarrow G_i^{-1}(1-\lambda^*/w_i)\) 
                                                    for \(i \ge l\)
        \STATE \(\alpha_i^* \leftarrow 0\) for \(i < l\)
    \ELSE
        \STATE \(\alpha_i^* \leftarrow G_i^{-1}(1-w_{l-1}/w_i)\)
                                                    for \(i \ge l\)
        \STATE \(\alpha_{l-1}^* \leftarrow c - \sum_{i=l}^n\alpha_i^*\)
        \STATE \(\alpha_i^* \leftarrow 0\) for \(i < l-1\)
    \ENDIF \label{alg:uncertainpi:stopsettingalpha}
\ENDIF
\end{algorithmic}
\end{algorithm}

The main intuition behind this algorithm is the fact that $\mc{V}_i(\alpha_i):=\EX_{p_i \sim G_i} V_i(\max(p_i-\alpha_i,0))$ is convex for linear $\{V_i\}$. So, we build on the KKT optimality conditions \cite{BoydV2004} of \eqref{eq:uncertain1} to design \UncertainPi~ 
 for computing the optimal allocation. 
The following lemma asserts the correctness of \UncertainPi. 

\begin{lemma}
\UncertainPi\ outputs $\{\alpha^*_i\}$ which solves \eqref{eq:uncertain1} for any strictly increasing \(\{G_i\}\), if the solution at
Step \ref{alg:lin:solvelambda} is obtained accurately.
\label{thm:uncertainpiworks}
\end{lemma}
Please see Appendix~\ref{sec:uncertainpiproof} for the proof of Lemma~\ref{thm:uncertainpiworks}.

For \(\{G_i\}\) strictly increasing over $\{[a_i,b_i]\}$,  \(\lambda^*\) at 
Step \ref{alg:lin:solvelambda} of \UncertainPi~can be obtained 
using binary search, whereas for special distributions, there exist closed forms.
For  uniform distribution, i.e., \(p_i \sim \text{Unif}[a_i,b_i]\) for all $i$:
\begin{align}
\lambda^* = \frac{\sum_{i=l}^n b_i - c}{\sum_{i=l}^n \frac{b_i-a_i}{w_i}}.
\label{eq:lambdauniform}
\end{align}

For uniform distribution we have the following unconditional correctness guarantee  and a bound on the computational complexity. This result follows from the above lemma and the fact that for uniform distributions, there exists a closed form solution to  Step \ref{alg:lin:solvelambda}.
\begin{theorem}
When \(\{p_i \sim \text{\em Unif}[a_i,b_i]\}\), \UncertainPi~finds the 
optimum of \eqref{eq:uncertain1} in \({O}(n^2)\).
\end{theorem}
\begin{proof}
For uniform distributions, the cumulative distribution functions \(G_i(\cdot)\) are linear, whose sum can be easily inverted. Computing the solution to Step \ref{alg:lin:solvelambda}, given by  \eqref{eq:lambdauniform}, 
has a time complexity of \({O}(n)\), and we need to compute this expression at most \(n\) times in Step \ref{alg:uncertainpi:untilconditionlambdasearch}.
Once we have found the correct \(\lambda\), computing the values of \(\{\alpha_i^*\}\) just requires a single pass over all the \(n\) users in Steps \ref{alg:uncertainpi:startsettingalpha}-\ref{alg:uncertainpi:stopsettingalpha} of Algorithm \ref{alg:uncertainpi}.
This ensures that the complexity of \UncertainPi\ is \({O}(n^2)\).
\end{proof}

For $\{p_i \sim \text{Unif}[a_i,b_i]\}$ it turns out that $\{\mc{V}_i\}$ are not strongly convex functions. So, among the gradient based methods, the best convergence bound for obtaining an $\epsilon$-accurate solution is $O\left(\frac{1}{\sqrt{\epsilon}}\right)$, which is due to the accelerated gradient descent algorithm \cite{Nesterov1998}. On the other hand, using a special structure of the problem (via KKT conditions), \UncertainPi~gives the exact optimum in $O(n^2)$ time. As the number of users per cell is in the range $10^1$--$10^2$, \UncertainPi~is computationally inexpensive.

\section{Simulation}
\label{sec:simulation}

\begin{figure}
\centering
\subfloat[]{
    \label{fig:known-pi}
    \includegraphics[width=0.9\columnwidth]{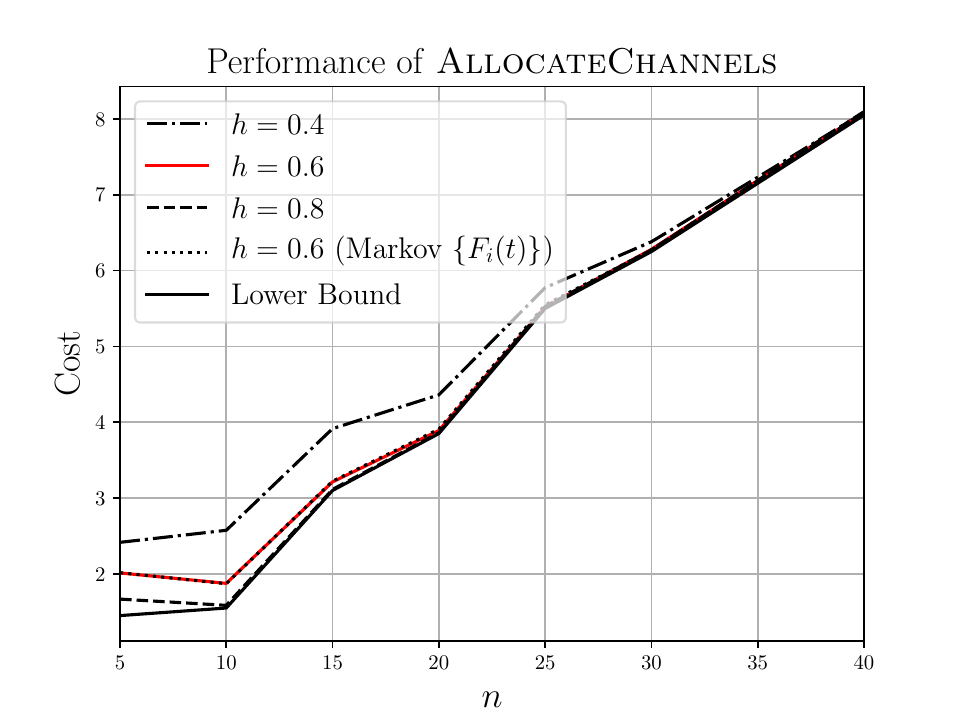}
}

\subfloat[]{
    \label{fig:unknown-pi}
    \includegraphics[width=0.9\columnwidth]{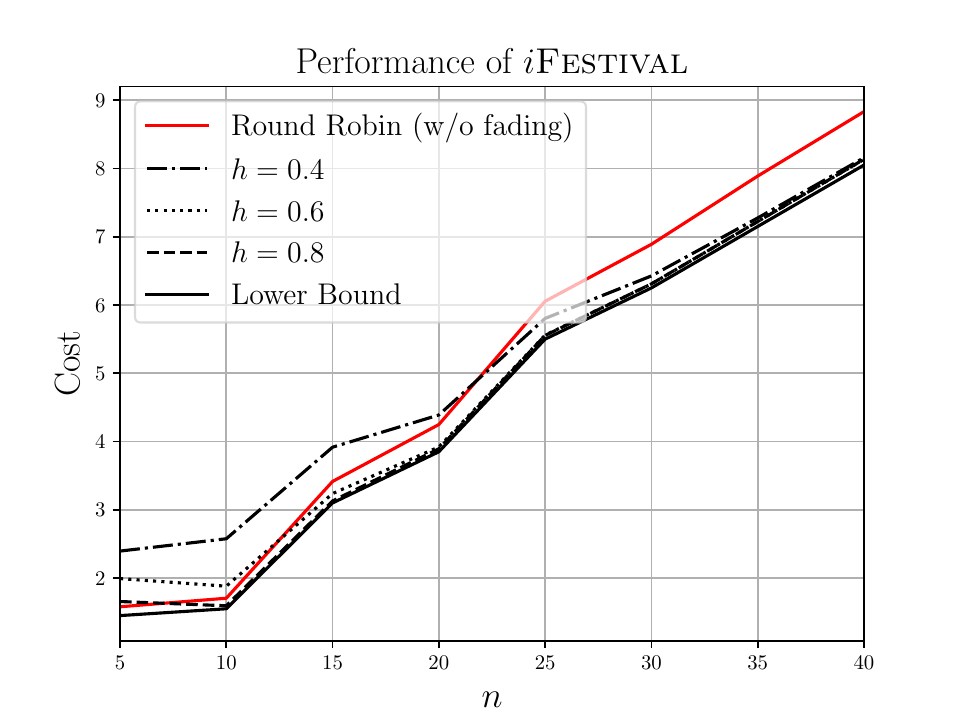}
}
    \caption{Performance of the algorithms.}
\end{figure}

For our simulations, we consider a system where  
the number of channels \(m\) scales as $0.4\ \times$ the number of users \(n\).
The mean consumption rates $\{p_i\}$ are drawn uniformly at random from the  set
$\{0.40, 0.45, 0.50, \ldots 0.80\}$.
The channels are assumed to be i.i.d. Bernoulli with ON probability  \(h\). For this overloaded system, we use  the following class of cost functions: $V_i(x) = p_i^{\theta} x^{1-\theta}$, for some $0<\theta<1$. 

In Fig. \ref{fig:known-pi}, we show the performance of \AllocateChannels, 
which requires us to know the mean consumption rates \(\{p_i\}\).
Here, we plot the asymptotic cost \({V}^{n,m}\) of running
\AllocateChannels\ under different fading conditions and compare
it with the lower bound.
For \(h=0.6\) and \(0.8\), \AllocateChannels\ is close to the lower bound at \(n=15\), and the cost almost matches the lower bound at \(n=20\).
Even for poor channel conditions (\(h=0.4\)), it matches the
lower bound at \(n=30\), i.e., $m=12$ channels.
We also observe almost the same performance when we change the consumption process  \(\{F_i(t)\}\) from i.i.d.
to Markov (with the same \(\{p_i\}\)).
This is expected since our theoretical  guarantees extend to any stationary and ergodic process.

In Fig. \ref{fig:unknown-pi}, we show the performance of \BufferBandit~which does not know the consumption statistics a priori and adapts as it learns those on the fly. 
For  different values of \(h\), we observe its performance to be close to \AllocateChannels~as well as the lower bound. We also compared its performance against a round robin schedule  run on a system that does not
experience fading (i.e., \(h=1\)).
Even under  heavy fading (\(h=0.4\)), for \(n\ge 25\), i.e., $m\ge 10$, \BufferBandit\ beats round robin's performance in the idealized scenario without fading.
This demonstrates that there is a significant benefit to learning
the consumption statistics using \BufferBandit\ and employing an algorithm such as
\AllocateChannels\ that optimally
utilizes the diversity of channel conditions.

\section{Quality degradation}
\label{sec:quality}

So far, we have addressed minimizing the cost due to buffering pauses when the network is overloaded, i.e., all users cannot  be supported at their minimum acceptable resolution levels. 
However, as discussed at the end of Sec.~\ref{sec:intro}, in an underloaded network, it is imperative to address user dissatisfaction due to quality degradation as well. In any network, the main objective then would be  mitigating buffering pause using {\em minimal} resource and then using the remaining resource to minimize user dissatisfaction due to quality degradation.

From the performance analysis in Thm.~\ref{thm:fading} and Sec.~\ref{sec:simulation}, it follows that \AllocateChannels~achieves almost zero frequency of pause using minimum number of channels $\BlockVariable \sum_i p_i$. Thus in an underloaded network, the first step would be to run \AllocateChannels~on $\BlockVariable \sum_i p_i$ channels and use the rest of the channels to minimize cost due to quality degradation.

For user $i$, let $q_i \in [0,1]$ be $\EX[F_i(t)]$ when the content is at the highest available resolution. For each user $i$, let  $W_i:\R_+ \to \R_+$ be a non-decreasing function. Then, in an underloaded network, the dissatisfaction of user $i$ due to quality degradation can be modeled as $W_i(q_i - \frac{\bar{s}_i}{\BlockVariable\mc{E}})$, where $\bar{s}_i$ is the ergodic service rate received by user $i$. So the problem of minimizing cost due to quality degradation using the remaining resource (after ensuring zero frequencies of  pause) is
\begin{align}
\min_{\{\bar{s}_i\ge p_i \BlockVariable\mc{E}\}}\sum_{i\in[n]} W_i(q_i - \frac{\bar{s}_i}{\BlockVariable\mc{E}}) \mbox{ s.t. } \sum_{i\in[n]} \frac{\bar{s}_i}{\BlockVariable\mc{E}} \le \frac{m}{\BlockVariable}. \nonumber
\end{align}

{ One may interpret  the cost $W_i(q_i - \frac{\bar{s}_i}{\BlockVariable\mc{E}})$ as a positive constant minus a utility that increases with the increased ergodic service rate. Following the intuition from traditional data networks where the utility saturates with increasing data rate, one may model $\{W_i\}$ as convex increasing functions. In that case the above problem is a standard convex optimization problem. On the other hand, if $\{W_i\}$ are modeled as concave increasing functions, a simple change of variables reduces it to \eqref{eq:benchmark}, and hence, can be solved using \ConcMin.} Thus, our work addresses both buffering pause and quality degradation.

\section{Conclusion}
\label{sec:discussion}
We study a resource allocation problem for minimizing user dissatisfaction due to buffering pauses during streaming over a multi-channel cellular network. Our consideration of a few previously overlooked practical aspects leads us to a novel continuous non-convex  problem with an interesting combinatorial structure. This problem is also related to learning in non-i.i.d. multi-armed bandits with delayed cost. We propose computationally efficient algorithms that are compatible with the current cellular implementations and provide theoretical guarantees for their (near) optimality.

\bibliographystyle{IEEEtran}
\bibliography{abrv,conf_abrv,multimedia}

\appendices
\section{Proofs of Lemma \ref{lem:benchmark} and Theorem \ref{thm:benchmark}}
\label{sec:benchmarkproof}

Let \(X_i(t) = \frac{Q_i(t)}{k\mc{E}}\), where \(Q_i(t)\) is the buffer evolution process defined in Sec. \ref{sec:model}.
Using assumption {\bf A3}, we can write the buffer evolution compactly as
\begin{align}
\label{eq:Xevolve}
X_i(t+1) = \left(X_i(t) + \frac{S_i(t)}{k\mc{E}} - F_i(t)\right)^+,
\end{align}
where \((\cdot)^+\) denotes \(\max(\cdot, 0)\).
Since we schedule in units of \(k\mc{E}\), we have \(\frac{S_i(t)}{k\mc{E}} \in \{0,1,\ldots\}\).

Following the discussion in Sec. \ref{sec:model}, using assumptions {\bf A1-A3}, we can express the probability of pause at time $t$ as
\begin{align*}
\EX\left[\mbf{1}\left(X_i(t) + \frac{S_i(t)}{k\mc{E}} - F_i(t) < 0\right)\right].
\end{align*}
Using the buffer evolution in Eq. \eqref{eq:Xevolve}, this can equivalently be written as
\begin{align}
\label{eq:kappai}
\EX\left[\mbf{1}\left(X_i(t+1) - \left(X_i(t) + \frac{S_i(t)}{k\mc{E}} - F_i(t)\right) > 0\right)\right].
\end{align}
This is because whenever \(X_i(t) + \frac{S_i(t)}{k\mc{E}} - F_i(t) \ge 0\),
\(X_i(t+1)\) would be equal to this expression and the argument of  the indicator in the above equation for \(\kappa_i\) would be \(0\). 
The only way for it to be positive is when \(X_i(t) + \frac{S_i(t)}{k\mc{E}} - F_i(t) < 0\).

Further, observe that since \(X_i(t) \in \{0,1,2,\ldots\}\), \(\frac{S_i(t)}{k\mc{E}} \in \{0,1,2,\ldots\}\), and \(F_i(t) \in \{0,1\}\), we have
\begin{align*}
X_i(t+1) - \left(X_i(t) + \frac{S_i(t)}{k\mc{E}} - F_i(t)\right) \in \{0,1\}.
\end{align*}
This implies that the indicator in Eq. \eqref{eq:kappai} is redundant as its
argument is always \(0\) or \(1\). 
So we get that the probability of pause at time $t$ is
\begin{align*}
\EX\left[X_i(t+1) - \left(X_i(t) + \frac{S_i(t)}{k\mc{E}} - F_i(t)\right)\right].
\end{align*}

When the buffer evolution is stationary and ergodic, i.e., \(\EX[F_i(t)] > \EX\left[\frac{S_i(t)}{k\mc{E}}\right]\), the processes are all stationary, and we have \(\EX[X_i(t+1)]=\EX[X_i(t)]\), and this gives us
\begin{align*}
\EX[F_i(t)] -  \EX\left[\frac{S_i(t)}{k\mc{E}}\right],
\end{align*}
which, by ergodicity, implies that 
$$\kappa_i = \EX[F_i(t)] -  \EX\left[\frac{S_i(t)}{k\mc{E}}\right].$$
Using the definition of \(\bar{s}_i\) in Lem. \ref{lem:benchmark}, and the definition of \(p_i\) in assumption {\bf A3}, we get
\begin{align*}
\kappa_i = p_i - \frac{\bar{s}_i}{k\mc{E}}
\end{align*}
which concludes our proof for the stationary and ergodic case.

When \(p_i \le \frac{\bar{s}_i}{k\mc{E}}\), the result follows by observing the drift of $\{X_t\}$ and the fact that if the expectations of a sequence of non-negative random variables upper-bounded by $1$ are $0$, then the sequence converges to $0$ almost surely.

\subsection{Proof of Theorem \ref{thm:benchmark}}
\label{sec:benchmarkproof:thm}

Let \(S_i^*(t)\) be the service under an optimal policy \(a^*\), and let the buffer evolution under such a policy be \(Q_i^*(t)\) for each user \(i\).
At any epoch, we have a total of \(m\mc{E}\) slots that can be scheduled, and this means
\begin{align*}
\sum_{i\in[n]} S_i^*(t) \le m\mc{E}
\end{align*}
for every epoch \(t\).

Since this hold for every epoch, the time average must satisfy this inequality as well, giving us
\begin{align*}
\sum_{i\in[n]}\bar{s}_i^* \le m\mc{E},
\end{align*}
where \(\bar{s}_i^*\) are the ergodic service rates under an optimal policy.
This implies that
\begin{align}
\label{eq:servicebound}
\sum_{i\in[n]}\frac{\bar{s}_i^*}{k\mc{E}} \le \frac{m}{k}.
\end{align}

Using Lem. \ref{lem:benchmark}, we get
\begin{align*}
V^{n,m}(a^*) = \sum_{i\in[n]} V_i\left(\max\left(p_i - \frac{\bar{s}_i^*}{m\mc{E}}, 0\right)\right).
\end{align*}

Since the optimal policy must satisfy Eq. \eqref{eq:servicebound}, the solution of the program in Thm. \ref{thm:benchmark} (Eq. \eqref{eq:benchmark}), can only have a lower value.
This gives us
\begin{align*}
V^{n,m}(a^*) \ge \bar{V}^{n,m}.
\end{align*}

\section{Proof of Theorem \ref{thm:concmin}}
\label{sec:concminproof}

First we shall prove that \ConcMin\ indeed finds the optimal service 
rates \(\{\alpha_i^*\}\). The optimization problem we are trying to solve
can be written as:
\begin{align}
\underset{\{\alpha_i\}}{\text{minimize}}\quad&
    \sum_i V_i (p_i - \alpha_i) \tag*{}\\
\text{subject to}\quad&
    \sum_i \alpha_i \le c \label{eq:sumlec}\\
\text{and}\quad&   
    0 \le \alpha_i \le p_i \quad \forall\ i \in [n]. \label{eq:alphadomain}
\end{align}

Recall that \(c=\frac{m}{k}\).
Since \(\{V_i\}\) are all concave functions, the optimal solution
happens at a corner point of the region defined by constraints
\eqref{eq:sumlec} and \eqref{eq:alphadomain}.
We have a total of \(2n+1\) linear inequations defining 
the feasible region (\(1\) in constraint \eqref{eq:sumlec} 
and \(2n\) in constraint \eqref{eq:alphadomain}).
Since there are \(n\) optimization variables \(\{\alpha_i\}\),
at every corner point, \(n\) of the inequations will hold with equality.
However, \(\alpha_i\) can't be equal to both \(0\) and \(p_i\), and so
at most \(n\) of the inequalities in constraint \eqref{eq:alphadomain} can
hold with equality.
As we just have one other constraint in \eqref{eq:sumlec}, we need at least
\(n-1\) of the constraints to hold with equality in constraint 
\eqref{eq:alphadomain}.
Therefore, in the optimal solution to the optimization problem,  there is 
at most one user who gets a non-zero rate but is not fully satisfied.

Let \(\mathcal{P} = \sum_ip_i\).
When \(\mathcal{P} \le c\), the optimal solution is trivial and we get
\(\alpha_i^* = p_i\) for all \(i\).
This case is handled in line \ref{alg:lin:plec} of \ConcMin.
Now consider the case \(\mathcal{P} > c\).
Let \(k^*\) be such that for all \(i \neq k^*\), either \(\alpha_i^*=0\) or
\(\alpha_i^* = p_i\) in the optimal solution \(\{\alpha_i^*\}\).
The preceding arguments guarantee that there is at least one such \(k^*\).
We find this \(k^*\) by looping over all of [n] in line 
\ref{alg:lin:findfracuser} of \ConcMin.
For each \(k \in [n]\), we find the optimal solution \(\{\alpha_i^k\}\) 
that satisfies, for all \(i \neq k\), \(\alpha_i^k = 0\) 
or \(\alpha_i^k = p_i\). 
Then we take the best among these over all values of \(k\).

When \(\mathcal{P} > c\), given a fixed \(k\), define 
\(S_k, Q_k \subseteq [n] \setminus k\) so that the ``optimal'' solution
\(\{\alpha_i^k\}\) satisfies the following properties:
\begin{align*}
\alpha_i^k &= 0 \quad \forall \ i \in S_k \\
\alpha_i^k &= p_i \quad \forall \ i \in Q_k \\
S_k \cap Q_k=\phi\quad &\text{and} \quad S_k \cup Q_k \cup \{k\} = [n]
\end{align*}
Since \(V_i(p_i) = V \cdot p_i\), \(S_k\) 
(or equivalently \(Q_k\)) can be found by solving
\begin{align*}
\underset{S_k}{\text{minimize}}&\quad 
    V\cdot\sum_{i \in S_k}p_i + 
            V_k\left(\mathcal{P}-c-\sum_{i \in S_k}p_i\right) \\
\text{subject to}&\quad
    \mathcal{P}-c-p_k \le \sum_{i \in S_k}p_i \le \mathcal{P}-c.
\end{align*}
The objective is a concave function of \(\sum_{i \in S_k}p_i\) and so the
minimum objective occurs at the maximum or minimum 
feasible value of \(\sum_{i \in S_k}p_i\).
We find \(\max_{S_k} \sum_{i \in S_k}p_i\) by solving 
\(\SSum([n] \setminus k, \mathcal{P}-c) (= R_k)\) on line 
\ref{alg:lin:findrk} of \ConcMin.
\(\min_{S_k} \sum_{i \in S_k}p_i\) subject to 
\(\sum_{i \in S_k}p_i \ge \mathcal{P}-c-p_k\) is the same as solving
\(\max_{Q_k} \sum_{i \in Q_k}p_i\) subject to \(\sum_{i \in Q_k}p_i \le c\).
This we do by \(\SSum([n] \setminus k, c)(=L_k)\) on line
\ref{alg:lin:findlk} of \ConcMin.
We then compare the costs of \(L_k\) and \(R_k\) to get the 
solution \(\{\alpha_i^k\}\) and the corresponding cost \(J_k\).

Observe that the optimal solution \(\{\alpha_i^*\}\) satisfies
\(\sum_i \alpha_i^* = c\) when \(\mathcal{P} > c\).
Also, we have \(\alpha_i^* = \alpha_i^{k^*}\) for some \(k^* \in [n]\).
Since we are comparing amongst feasible solutions \(\{\alpha_i^k\}\)
in line \ref{alg:lin:findkstar} of \ConcMin, we get the optimal \(k^*\)
and hence the optimal solution \(\{\alpha_i^*\}\).
This shows that \ConcMin\ outputs the optimal solution.

See Sec. \ref{sec:ConcMin} for a discussion on the computational complexity of \ConcMin.

\section{Proof of Theorem \ref{thm:fading}}
\label{sec:fadingproof}

The proof of Theorem \ref{thm:fading} follows from the following two observations: (i) the expected number of slots  given to user \(i\) by \SelectUsers\ is \(\alpha_i^*\), and (ii) there exists a perfect matching between the selected users and channels (having the highest fading state or rate $1$)  with very high probability.
We state these as two lemmas.

\begin{lemma}
\label{lem:selectgivesalpha}
Let \(N_i(t)\) be the number of times user \(i\) appears in the list selected by \SelectUsers\ at epoch \(t\). Then \(\EX[N_i(t)]=\alpha_i^*\).
\end{lemma}
\begin{proof}
Since \(\alpha_i^* \le p_i \le 1\) for all \(i \in [n]\), and the each slot has an interval of size \(1\), a user can appear for at most two slots (see Figure \ref{fig:SelectUsers}).
If the user's \(\alpha_i^*\) occupies only one slot, then the lemma follows directly since the user gets selected for that slot with probability \(\alpha_i^*\) and for no other slot.
If the user occupies two slots, then user gets selected for some slot \(s_j\) with probability \(\alpha_a\) and for \(s_{j+1}\) with probability \(\alpha_b\) such that \(\alpha_a+\alpha_b=\alpha_i^*\).
Since expectation is a linear operator, we get the lemma.
\end{proof}

\begin{lemma}
\label{lem:perfectmatching}
For the bipartite graph \(G = (L \cup R, E)\) described in Algorithm \ref{alg:AllocateChannels} (\AllocateChannels), 
\begin{align*}
\PR(G\ \text{has no perfect matching}) \le \theta^{-m},
\end{align*}
for some constant  \(\theta > 1\) and a large enough \(m\). 
Here \(L\) is the set of nodes corresponding to the list of selected users, and \(R\) is the set of channels.
An edge \((l,r) \in E\) iff the  channel \(r\) is ON for user \(l\).
\end{lemma}
\begin{proof}
Proof of this lemma follows along the lines of the proof of \cite[Lemma\ 1]{bodas14}.
The key idea is Hall's theorem, which states that for any bipartite graph \(G = (L \cup R, E)\) which does not have a perfect matching, there exists a set \(A \subseteq L\) whose neighborhood is smaller than itself, i.e., \(|\Gamma(A)| < |A|\), where
\begin{align*}
\Gamma(A) = \{r\ \mid\ \exists\ l \in L\ \text{such that}\ (l,r) \in E\}
\end{align*}
is the neighborhood of \(A\) (see \cite{bodas14} and the references therein).

Let \(a=|A|\).
For \(|\Gamma(A)| < a\), we need at least \(m-a+1\) channels to not be OFF for all the elements in a.
\(A\) contains at least \(\left\lceil\frac{a}{2}\right\rceil\) distinct users since no user can appear more than twice in \(L\).
The probability that a particular subset of \(R\) of size \(m-a+1\) has no ON connection to any element of \(A\) is therefore upper bounded by \((1-\bar{h})^{(m-a+1)\lceil a/2 \rceil}\).
Taking union bound over all sets of channels of size \(m-a+1\), we get
\begin{align*}
\PR(|\Gamma(A)| < |A|) &\le {m \choose m-a+1}(1-\bar{h})^{(m-a+1)\lceil a/2 \rceil} \\
&\le {m \choose m-a+1}\left(\sqrt{1-\bar{h}}\right)^{(m-a+1)a}.
\end{align*}

Further taking union bound over all non-empty subsets of \(L\), we get
\begin{align*}
\PR(G\ \text{has no perfect matching}) &\le 
    \sum_{a=1}^m {m \choose a}{m \choose m-a+1}\delta^{(m-a+1)a}
\end{align*}
where \(\delta = \sqrt{1-\bar{h}} < 1\).
This gives us
\begin{align*}
\PR(\text{No perfect matching}) &\le 
    2\sum_{a=1}^{\lceil m/2 \rceil}{m \choose a}{m \choose m-a+1}\delta^{(m-a+1)a} \\
&= 2\sum_{a=1}^{\lceil m/2 \rceil}{m \choose a}{m \choose a-1}\delta^{(m-a+1)a} \\
&\le 2\sum_{a=1}^{\lceil m/2 \rceil}m^{2a}\delta^{ma/2},
\end{align*}
where the last inequality follows from \({m \choose a} \le m^a\), \({m \choose a-1} \le m^a\), and \(m-a+1 \ge \frac{m}{2}\) for \(a\) in \(1,  \ldots, \left\lceil\frac{m}{2}\right\rceil\).

For a large enough \(m\), \(\left(m^2\delta^{m/2}\right) < 1\), and so \(\left(m^{2a}\delta^{ma/2}\right)\), has its maximum at \(a=1\).
This gives us
\begin{align*}
\PR(G\ \text{has no perfect matching}) &\le 
    2m \times m^2\delta^{m/2}.
\end{align*}
We can always find a \(\theta > 1\) such that for large enough \(m\), \(\left(2m^3\delta^{m/2}\right) \le \theta^{-m}\).
For example, set \(\theta = \frac{2}{1+\delta}\).
Since \(\delta < 1\), this gives us \(\theta > 1\), and concludes our proof.
\end{proof}

Now we are in a position to prove Thm. \ref{thm:fading}.
Using the assumptions {\bf A1-A3}, we get 
\begin{align*}
\EX&\left[S_i^\AllocateChannels(t)\right] \ge \\
&\EX\left[S_i^\AllocateChannels(t)\ \mid\ \text{we find a perfect matching at}\ t\right]\ \times \\
&\qquad\qquad\qquad\qquad\qquad\qquad \PR(\text{we find a perfect matching at}\ t) \\
=&\ k\mc{E}\alpha_i^*\ \times\ \PR(\text{we find a perfect matching at}\ t),
\end{align*}
where the last equality follows from Lem. \ref{lem:selectgivesalpha}.
Using Lem. \ref{lem:perfectmatching}, we get
\begin{align*}
\bar{s}_i^\AllocateChannels \ge k\mc{E}\alpha_i^*\ \times\ (1-\theta^{-m}).
\end{align*}
Since the outputs of \ConcMin, \(\{\alpha_i^*\}\) satisfy \(0 \le \alpha_i^* \le p_i\), this gives us
\begin{align*}
\kappa_i^\AllocateChannels \le p_i - \alpha_i^* + \alpha_i^* \theta^{-m}.
\end{align*}
Using assumption {\bf A1}, we get
\begin{align*}
V_i(\kappa_i^\AllocateChannels) \le V_i(p_i-\alpha_i^*) + G\alpha_i^*\theta^{-m},
\end{align*}
or
\begin{align*}
\sum_{i\in[n]} V_i(\kappa_i^\AllocateChannels) &\le \sum_{i\in[n]}V_i(p_i-\alpha_i^*) + \left(G\theta^{-m}\sum_{i\in[n]}\alpha_i^*\right) \\
&\le \bar{V}^{n,m} + Gm\theta^{-m}.
\end{align*}

For a large \(m\), we can always find a \(\gamma\) such that \(Gm\theta^{-m}\le\gamma^{-m}\) for all \(n\). 
For example, use \(\gamma=\frac{1+\theta}{2}\).
Since \(\theta>1\), we get \(\gamma>1\), and thus for a large \(m\),
\begin{align*}
V^{n,m}(\AllocateChannels)-\bar{V}^{n,m} \le \gamma^{-m}.
\end{align*}

\section{Proof of Theorem \ref{thm:bandit}} \label{app:bandit}
Over a time horizon $T$ the total regret can be divided into the following parts according to phases: regret over phases $r^q$ to $r^{q+1}$ for $q=0$ to $\lfloor \log_r T\rfloor-1$ and the regret over the remaining epochs till $T$.

By simple concentration inequality for an i.i.d. Bernoulli process and assumption {\bf A2}, the probability that the estimates of all $p_i$ are correct after phase $r^q$ is upper-bounded by 
$2n \exp(-\frac{\min_{i\neq j} |p_i-p_j|^2}{2} q)$.  Let us first bound the regret assuming that the $p_i$ estimates are correct.

Let us first consider the case without fading, i.e., $\mbf{H}(t)=\mbf{1}$. In this case, consider  for any $i$ with $\alpha^*_i>0$
$\alpha_i=\alpha^*_i-\delta$ for some $\delta>0$. By coupling the arrival into the original queue with that of this concocted dynamics one can directly argue that the expected number of pauses in the original dynamics is upper bounded by that of this dynamics. So, it is enough to bound the expected number of pauses in this concocted dynamics. 

The expected number of pauses for user $i$ till time $t$ is upper bounded by the expected duration for which its buffer stays empty between time $0$ and $t$ times $p_i(1-\alpha_i)$.. 
Note that the duration for which the buffer stays empty can be divided into phases of the algorithm. Further, for obtaining an upper bound one can assume that the buffer restarts from $0$ at the beginning of every phase. This again follows using an elementary coupling.

Using Proposition 4.1 of \cite{HerveL2014}, for any $i$ with $\alpha_i^*>0$ the expected duration the  buffer stays  empty during a phase, given $p_i$ estimates are correct,  is upper bounded by

\[w r^q(r-1) \pi_i(0) + O(1/\delta^2),\]
where $\pi_i(0)$ is the stationary probability of the concocted Markov chain to be at $0$. This follows by considering the transitions of the concocted Markov chain and computing the right parameters in \cite[Prop.~4.1]{HerveL2014}.

As the concocted chain is a lazy birth death chain, it follows that $\pi_i(0)$ is  $\frac{p_i-\alpha_i}{p_i(1-\alpha_i)}$, for $\alpha_i<p_i$. 

Hence, the expected number of buffering pauses for user $i$ over a time horizon $T$ can be upper-bounded by

\[T  \pi_i(0) p_i~(1-\alpha_i) + O(1/\delta^2) (\log_r(T)+1).\]

As $\alpha_i*\le p_i$, this, in turn, can be written as

\[ T \max(p_i-\alpha_i^*, 0) + T  (\alpha_i^*-\alpha_i) + (\log_r(T)+1)~O(1/\delta^2).\]

Proof of the case with $\mbf{H}(t)=\mbf{1}$ is completed by combining the above bound with the fact that  $p_i$ estimates can be wrong with probability no more than $2n \exp(-\frac{\min_{i\neq j} |p_i-p_j|^2}{2} q)$. The final bound follows by choosing the right constants mentioned in the theorem and $\delta=\frac{1}{T^{1/3}}$.

For the i.i.d. fading case, note that  the dynamics of $\{Q_i(t)\}$ is same as  that of the concocted Markov chain in the no fading case with $\delta=\theta^{-m}$.  This is because  in the fading case we derived (Appendix \ref{sec:fadingproof}, after Lemma~\ref{lem:perfectmatching}) that under our proposed \AllocateChannels~$\alpha_i \ge \alpha_i^*-\theta^{-m}$. The result follows by plugging in the  parameter values mentioned in the theorem.

\section{Proof of Lemma \ref{thm:uncertainpiworks}}\label{sec:uncertainpiproof}
The problem we are trying to solve can be written as:
\begin{align}
\underset{\{\alpha_i\}}{\text{minimize}}\quad& \mathbb{E}\left[ 
                \sum_i w_i \left(p_i - \alpha_i\right)^+ \right] 
                \tag{UPI} \label{prog:uncertainpi}\\
\text{subject to}\quad& \alpha_i \ge 0 
                \quad \forall\ i \in [n], \label{eq:alphapositive} \\
& \alpha_i \le b_i 
                \quad \forall\ i \in [n], \label{eq:alphasensible} \\
& \sum_i \alpha_i \le c. \label{eq:alphalec}
\end{align}

Let \(\{\alpha_i^*\}\) be the optimal solution to this program.
Recall that \(g_i(x)\) is non-zero iff \(x \in [a_i, b_i]\).
Further, partition the set \([n]\) into the following:
\begin{align*}
P = \{i \mid \alpha_i^* = 0\}, \
Q = \{i \mid 0 < \alpha_i^* < a_i\}, \
R = \{i \mid a_i \le \alpha_i^* < b_i\}, \
S = \{i \mid \alpha_i^* = b_i\}.
\end{align*}

Let the KKT multipliers for constraints
\eqref{eq:alphapositive}, \eqref{eq:alphasensible}, and \eqref{eq:alphalec}
be \(\theta_i\), \(\phi_i\), and \(\lambda\) respectively.
\ref{prog:uncertainpi} is clearly a convex program and writing down the 
KKT conditions and simplifying them gives us:
\begin{align}
\lambda &= w_i + \theta_i \quad \forall\ i \in P, \label{eq:lambdaP}\\
\lambda &= w_i \quad \forall\ i \in Q, \label{eq:lambdaQ}\\
\lambda &= w_i (1 - G_i(\alpha_i^*)) \quad \forall\ i \in R, \label{eq:lambdaR}\\
\lambda &= -\phi_i \quad \forall\ i \in S. \label{eq:lambdaS}
\end{align}

We can see that if \(S\) is non-empty, we have \(\lambda = 0\) since all the KKT multipliers have to be non-negative, which 
implies that \(P\), \(Q\), and \(R\) are empty.
This gives us the case when \(\sum_i b_i \le c\) and hence
\(\alpha_i^* = b_i\) for all \(i\).

When \(S\) is
empty, \(\lambda\) acts as a threshold for giving
users a non-zero rate: if \(w_i < \lambda\), user \(i\) gets a rate \(0\),
and if \(w_i > \lambda\), user \(i\) gets a rate \(\alpha_i^* \ge a_i\).
Among the users where \(\lambda = w_i\), we can divide the rate left over 
after allocating to users with a higher \(w_i\) in any way.
Using the fact that \(w_1 < w_2 < \ldots < w_n\), \UncertainPi\ first finds this \(\lambda\) consistent with Equations \eqref{eq:lambdaP}, \eqref{eq:lambdaQ}, and \eqref{eq:lambdaR} in Steps \ref{alg:uncertainpi:startlambdasearch}-\ref{alg:lin:solvelambda} of Algorithm \ref{alg:uncertainpi}.
Then \UncertainPi\ computes the optimal rates are computed using Equations \eqref{eq:lambdaP} and \eqref{eq:lambdaR} in Steps \ref{alg:uncertainpi:startsettingalpha}-\ref{alg:uncertainpi:stopsettingalpha}.
Any remaining rate is given to the user satisfying \eqref{eq:lambdaQ}.
Since the rates \(\{\alpha_i^*\}\) we get this way satisfy the KKT conditions, they are an optimal solution for the program \ref{prog:uncertainpi}.

\end{document}